\documentclass[aps,pra,twocolumn,groupedaddress]{revtex4}

\usepackage[english]{babel}
\usepackage[dvips]{graphicx}
\usepackage{enumerate,amsthm,amsmath,amssymb,color,graphicx,bbm,float}
\usepackage{natbib}

\newtheorem{theo}{Theorem} 
\newtheorem{lemma}[theo]{Lemma}

\begin{document}

\title{A simple proof that Gaussian attacks are optimal among collective attacks against continuous-variable quantum key distribution with a Gaussian modulation.}

\author{Anthony Leverrier}
\affiliation{Institut Telecom / Telecom ParisTech, CNRS LTCI,\\
  46, rue Barrault, 75634 Paris Cedex 13, France} \author{Philippe
  Grangier}
\affiliation{Laboratoire Charles Fabry, Institut d'Optique, CNRS, Univ. Paris-Sud,\\
  Campus Polytechnique, RD 128, 91127 Palaiseau Cedex, France}

\date{\today}

\begin{abstract}
In this paper, we give a simple proof of the fact that the optimal collective attacks against continous-variable quantum key distribution with a Gaussian modulation are Gaussian attacks. Our proof, which makes use of symmetry properties of the protocol in phase-space, is particularly relevant for the finite-key analysis of the protocol, and therefore for practical applications.
\end{abstract}


\maketitle

\section{Introduction}

Quantum key distribution (QKD) is a cryptographic primitive allowing two distant parties, traditionally referred to as Alice and Bob, to establish a secret key \cite{SBC08}. This key can later be used to secure sensitive communication thanks to one-time pad for instance. QKD has received a lot of attention lately as it is the first application of quantum information science which could be developed on a large scale. For instance, metropolitan networks are certainly compatible with present technology, as was recently demonstrated in Vienna with the SECOQC project \cite{PPA09}. 

Historically, QKD protocols have been using discrete variables, meaning that Alice and Bob exchange information encoded on a finite-dimensional Hilbert space such as the polarization of a single photon for instance. Hence, protocols such as BB84 \cite{BB84} have been studied for a long time and their unconditional security is today well established \cite{ren05}, at least in a scenario where side-channels are not considered \cite{SK09}.

More recently, it was suggested that one could encode information on continuous variables in phase-space to perform QKD \cite{ral99}. Practical schemes requiring only coherent states together with an homodyne detection were introduced by Grosshans and Grangier in 2002 (GG02), first with direct \cite{GG02} and then with reverse \cite{GG02b} reconciliation, and later successfully implemented \cite{GVW03,FDD09}. These protocols were proven secure against collective attacks \cite{GC06,NGA06}, which are optimal in the asymptotic limit \cite{RC09}. Let us recall that the optimal collective attacks are Gaussian attacks, meaning that the eavesdropper operation corresponds to a Gaussian map.

The basic idea of the protocol GG02 is the following: Alice draws two random numbers $q_A$ and $p_A$ with a Gaussian probability distribution and sends the coherent state $|q_A + i p_A\rangle$ to Bob. Bob chooses a random quadrature and performs an homodyne detection for that quadrature: he then obtains the classical variable $y$, a noisy version of either $q_A$ or $p_A$. He finally informs Alice of his choice of quadrature. Alice keeps her relevant classical variable which she notes $x$. Repeating this operation $n$ times, Alice and Bob end up with two correlated vectors ${\bf x} = (x_1,\cdots, x_n)$ and ${\bf y} = (y_1,\cdots, y_n)$ from which they can distill a secret key by applying the usual classical post-processing composed of parameter estimation, error reconciliation and privacy amplification. Note that a small variation of this protocol consists in performing an heterodyne detection on Bob's side instead of an homodyne detection \cite{WLB04}. The security of this variant was investigated in \cite{LG07, SMG07} where the optimal individual attack is explicited.

Other variations of this GG02 protocol consist in replacing the Gaussian modulation with a discrete modulation \cite{NH03,NH04,HL06,ZHR09,SL09,LG09,LG10}, or adding a post-selection procedure to the protocol \cite{SRL02,LSS05,NH05,NH06,HL07}.

One main advantage of the protocols with a Gaussian modulation but without post-selection is that they display a high level of symmetry. In particular, a specific symmetry of these protocols in phase-space was recently investigated in \cite{LKG09} and appears to be a good approach in order to improve the known lower bounds of the secret key rate against arbitrary attacks \emph{in the finite size regime}. Remember that Ref. \cite{RC09} proves that collective attacks are optimal in the asymptotic regime thanks to a de Finetti-type theorem which gives rather conservative bounds when finite size effects are taken into account. A general framework for the finite size analysis of QKD was developped in \cite{SR08} and the first numerical results appear to be rather pessimistic \cite{CS09}, hence giving incentive to improve known bounds, in particular with the help of symmetries. Partial results in this direction, such as a de Finetti-type theorem in phase-space, were already obtained in \cite{LC09}. Whereas in \cite{LKG09}, the authors examined the possibility to use the specific symmetries of GG02 to prove the security of the protocol against general attacks, our goal here is more modest as we show that these symmetries allow one to easily recover known results concerning the optimality of Gaussian attacks among all collective attacks. A novelty of our proof compared to previous techniques \cite{GC06,NGA06} is that it can be applied in the finite size scenario.

\section{A new security proof against collective attacks}

The main idea of our proof is to use symmetries of the protocol to simplify the analysis of its security. In general, the security of a usual Prepare and Measure protocol where Alice prepares and sends quantum states to Bob (coherent states with a Gaussian modulation in the case of GG02) is analysed through an equivalent entangled version of the protocol. For GG02, this entangled version consists for Alice in preparing two-mode squeezed vacuua, measuring one mode of these states with an heterodyne detection and sending the other mode to Bob through the quantum channel \cite{GCW03}.

The security of the entangled protocol is then analysed through the $n$-mode bipartite quantum state $\rho_{AB} \in \left( \mathcal{H}_A \otimes \mathcal{H}_B \right)^{\otimes n}$ shared by Alice and Bob before they perform their measurements. Here, $\mathcal{H}_A$ and $\mathcal{H}_B$ refer respectively to Alice and Bob's single mode Hilbert spaces. Unfortunately, the total Hilbert space $\left( \mathcal{H}_A \otimes \mathcal{H}_B \right)^{\otimes n}$ is usually too big to allow for a complete analysis. 

A solution is therefore to use specific symmetries of the protocol in order to show that only a symmetric subspace of $\left( \mathcal{H}_A \otimes \mathcal{H}_B \right)^{\otimes n}$ needs to be considered. Indeed, one can show that if a QKD protocol is invariant under a certain class of symmetries, say invariance under permutation of the subsystems of Alice and Bob, then one can safely assume that the quantum state $\rho_{AB}$ displays the same symmetry.

This might look a bit suspicious at first sight as one may object that the eavesdropper is free to break the symmetry of the state, hence invalidating the previous statement. The way to solve this apparent paradox is to recall that, without loss of generality, one can always assume that Eve is given a purification $|\psi\rangle_{ABE}$ of $\rho_{AB}$. Since the protocol is invariant under the group of symmetry $\mathcal{G}$, Alice and Bob can consider the state $\bar{\rho}_{AB}$ which is obtained by averaging their initial state $\rho_{AB}$ over the group $\mathcal{G}$. As far as Alice and Bob are concerned, applying the QKD protocol (measurements, parameter estimation, reconciliation and privacy amplification) to the state $\bar{\rho}_{AB}$  is indistinguishable from applying it to the state $\rho_{AB}$ . Now, because the state $\bar{\rho}_{AB}$ is invariant under the action of $\mathcal{G}$, it is possible to find a purification $|\bar{\psi}\rangle_{ABE}$ of this state such that $g |\bar{\psi}\rangle_{ABE} = |\bar{\psi}\rangle_{ABE}$ for all $g \in \mathcal{G}$. This was proven in the case of the symmetric group $\mathcal{S}_n$ in \cite{ren05} and in the case of locally compact groups in \cite{CKR09}. Then it is shown in \cite{CKR09} that there exists a completely positive trace-preserving map $\mathcal{T}$ mapping $|\bar{\psi}\rangle_{ABE}$ to $|\psi\rangle_{ABE}$. Hence, the eavesdropper has at least as much information when her state corresponds to the symmetric purification $|\bar{\psi}\rangle_{ABE}$ as when her state corresponds to the (non necessary symmetric) purification $|\psi\rangle_{ABE}$. This means that considering the state $|\bar{\psi}\rangle_{ABE}$ is sufficient to evaluate the security of the protocol. As a conclusion, Alice and Bob can always assume that their bipartite state displays the same symmetry properties as the QKD protocol.

In addition to use specific symmetries of the protocol, one can simplify the security analysis further by restricting the eavesdropper's action to a certain class of attacks, for instance, \emph{collective attacks}. 
This means that the bipartite quantum state shared by Alice and Bob is assumed to be independent and identically distributed (i.i.d.), that is, that there exists a probability distribution $p(\sigma_{AB})$ on $\mathcal{H}_A \otimes \mathcal{H}_B$ such that:
\begin{equation}
\rho_{AB} = \int \sigma_{AB}^{\otimes n} p(\sigma_{AB}) {\rm d} \sigma_{AB}.
\end{equation}

In the case of protocols such as BB84 which are invariant under permutation of Alice and Bob's subsystems, it is useless to consider symmetries of the protocol when considering collective attacks since an i.i.d. state is clearly invariant under permutation of its subsystems. The converse property is not true in general. 
However, the exponential version of de Finetti theorem \cite{ren07} and the post-selection technique introduced in \cite{CKR09} show that it also holds asymptotically. 

In the case of continuous-variable QKD protocols, one can consider a specific symmetry in phase-space \cite{LKG09} which is not strictly implied by collective attacks.  The protocol GG02 is indeed invariant under conjugate passive symplectic operations applied by Alice and Bob. Physically, this invariance means that the protocol is not affected when Alice processes her $n$ modes into any passive linear interferometer while Bob processes his $n$ modes into the passive linear interferometer effecting the conjugate orthogonal transformation in phase space.
To see this, it is enough to show that the reconciliation procedure as well as the parameter estimation would perform equally well whether or not conjugate passive symplectic operations are applied. Let us consider first the reconciliation procedure which consists in turning Alice and Bob's measurement results into a identical bitstrings. Such a procedure (see Ref. \cite{LAB08} for a specific example) is designed to work in the case where Alice's classical data follow a Gaussian modulation and the correlation between Alice and Bob's data are is measured by their covariance. Since passive symplectic operations in phase space correspond to orthogonal transformations for Alice and Bob's measurement results, neither the Gaussian modulation nor the covariance of the data are affected, which guarantees that the reconciliation procedure is transparent to such transformations.
Concerning the parameter estimation, which is used in particular to compute Eve's information, it is notable that for the protocol GG02, only the covariance matrix of the state $\rho_{AB}$ should be estimated, and more specifically the \emph{transmission} and \emph{excess noise} of the quantum channel. Both these quantities are invariant under any orthogonal transformation of the data.
This means that the state $\rho_{AB}$ can safely be considered to be invariant under conjugate passive Gaussian operations appied by Alice and Bob.

Using this symmetry together with the assumption of collective attacks leads to a simple proof that the optimal collective attacks are Gaussian. More precisely, if the adversary is restricted to perform a collective attack, Alice and Bob can safely assume that this attack is Gaussian. To show this, it is enough to prove that the state $\rho_{AB}$ can be considered Gaussian. Indeed, at the beginning of the protocol, Alice prepares $n$ two-mode squeezed states, which is a $2n$-mode Gaussian state. If the quantum state shared by Alice and Bob at the end of the protocol is also Gaussian, it means that the quantum channel can be described as a Gaussian map. Our proof is based on the following Lemma.

\begin{lemma}
\label{lemma}
If a bipartite $2n$-modal quantum state $\rho_{AB}$ (for $n \geq 2$) is both i.i.d. and invariant under conjugate passive Gaussian operations, then $\rho_{AB}$ is a Gaussian state.
\end{lemma}

\begin{proof}
Let us first rephrase the lemma in phase-space representation. Any state $\rho_{AB}$ is completely characterized by its Wigner function $W_{\rho}(x, p, y, q)$ where $x, p$ are $n$-dimensional vectors corresponding to Alice's phase-space and $y, q$ correspond to Bob's phase-space. The application of a passive Gaussian operation on Alice's modes and of its conjugate operation on Bob's modes maps the state $\rho$ to the state $\rho'$. The Wigner function $W_{\rho'}(x, p, y, q)$ of $\rho'$ is equal to $W_{\rho}(x', p', y', q')$ for the change of coordinates  $(x', p', y', q') = S^T (x, p, y, q)$ and the symplectic map $S$ can be written as
\begin{equation}
\label{symplectic}
S = S(X,Y) \equiv
\begin{pmatrix}
X & Y & 0 & 0 \\ 
-Y & X & 0 & 0 \\
0 & 0 & X^T & -Y^T\\  
0 & 0 & Y^T & X^T\\  
\end{pmatrix} 
\end{equation}
where the matrices $X$ and $Y$ are such that \cite{ADMS95}:
\begin{eqnarray}
X^T X+Y^T Y &=& X X^T + Y Y^T =1\\
X^T Y &,& X Y^T \quad \mathrm{symmetric}.
\end{eqnarray}
In order to prove the lemma, we observe that if any such map $S$ leaves the Wigner function invariant, then $W$ can only depend on three parameters which are $||x||^2+||p||^2$, $||y||^2+||q||^2$ and $x \cdot y - p \cdot q$ (a proof of this fact can be found in Appendix \ref{proof}). This means that there exists a function $f: \mathbb{R}^+ \times \mathbb{R}^+ \times \mathbb{R} \mapsto \mathbb{R}$ such that:
\begin{multline}
W_{\rho}(x, p, y, q) \\
= f(||x||^2+||p||^2, ||y||^2+||q||^2, x \cdot y - p \cdot q).
\end{multline}
Then, since $\rho_{AB}$ is an i.i.d. state, the same must be true for $f$, meaning in particular that
\begin{multline}
f\left(\sum_{i=1}^n x_i^2 + p_i^2, \sum_{i=1}^n  y_i^2 + q_i^2,\sum_{i=1}^n  x_i y_i - p_i q_i\right) \\
\propto \prod_{i=1}^n f(x_i^2 + p_i^2, y_i^2 + q_i^2, x_i y_i - p_i q_i),
\end{multline}
which is exactly the characterization of the exponential function. Hence, $f$ and also $W$ are exponential in $||x||^2+||p||^2$, $||y||^2+||q||^2$ and $x \cdot y - p \cdot q$, which means that the state $\rho_{AB}$ is a Gaussian state. This concludes our proof.
\end{proof}

The protocol GG02 is invariant under conjugate passive symplectic operations applied by Alice and Bob. 
Hence Alice and Bob can safely assume that their state $\rho_{AB}$ displays the same symmetry. 
Restricting the analysis to collective attacks, one can use Lemma \ref{lemma} to conclude that the state $\rho_{AB}$ can be considered to be Gaussian.
Since the inital state produced by Alice, a (Gaussian) two-mode squeezed vacuum is transformed through the quantum channel into another Gaussian state, this means that the action of the channel, that is of the attack, can be safely considered to be Gaussian, which gives a simple proof that Gaussian attacks are optimal among collective attacks.

\section{Conclusion and perspectives}

In this paper, we gave an alternative proof that Gaussian attacks are optimal against GG02 among all collective attacks. This new proof makes use of symmetries of the protocol in phase-space, and does not require to consider specific properties of the entropy as in previous proofs \cite{GC06,NGA06}. 
A natural question is whether this technique can be exploited for variants of the GG02 protocol.

Let us consider first protocols with a discrete modulation, such as \cite{LG09}. In this case, our new proof cannot be applied directly as protocols with a discrete modulation are less symmetric than protocols with a Gaussian modulation. Indeed, not all rotations in phase-space leave the protocol invariant: only the orthogonal transformations leaving the modulation unchanged, that is, transformations belonging to the symmetry group of the hypercube are relevant in this case. This group, however, is much smaller that the group considered here, and one cannot conclude directly that the state $\rho_{AB}$ can be safely considered to be Gaussian. Note that this is still true but has to be proven with a different approach \cite{LG09,LG10} based on the extremality of Gaussian states \cite{WGC06}.

The second class of protocols one could consider is protocols with a post-selection procedure \cite{SRL02,LSS05,NH05,NH06,HL07}. These protocols have not yet be proven secure against general collective attacks because it is not known whether Gaussian attacks are optimal among collective attacks. The technique presented in this paper cannot be used either for protocols displaying a post-selection step as this post-selection explicitly breaks the symmetry of the protocol in phase-space. 

In addition to its simplicity, our new proof turns out to be particularly useful for the finite size analysis of the security of continuous-variable QKD protocols. Indeed, a specificity  of the finite size analysis is that Alice and Bob cannot assume to perfectly know the quantum state they share. 
For continuous-variable protocols in general, this is in fact theoretically impossible as their state belongs to an infinite dimensional Hilbert space, and therefore requires an infinite number of parameters to be fully described. Fortunately, for protocols such as GG02 where the state can safely be considered to be Gaussian, Alice and Bob only need to know their covariance matrix which depends on three parameters: the modulation variance which is chosen by Alice as well as the transmission and the excess noise of the quantum channel. 
These parameters are estimated by revealing part of Alice and Bob's data. In order to proceed with this estimation, one needs a statistical model and choosing a normal model seems quite natural. However, previous proofs of Gaussian optimality presented in \cite{GC06, NGA06} assume that the covariance matrix is known from Alice and Bob and cannot justify the use of a normal statistical model for its estimation. The proof presented here, on the contrary, allows for such a justification (see Appendix \ref{normal_model} for details). 

The fact that our proof applies to finite size analysis is crucial as our ultimate goal is clearly to assess the security of practical implementations, which are necessary finite. A general finite-size analysis of continuous-variable protocols will be the subject of future work.

\appendix

\section{Complete proof of Lemma \ref{lemma}}
\label{proof}

Before considering the general case of Wigner functions, let us first consider the case of a probability distribution $p(x,y)$ which is invariant under orthogonal transformations applied to both $x$ and $y$. In other words, for any $R \in O(n)$, one has $p(Rx,Ry)=p(x,y)$. Such a symmetry property clearly implies that $p(x,y)$ can only depend on three parameters, namely $||x||, ||y||$ and $x\cdot y$. With Wigner functions, the argument is more subtle, and is detailed below.

We want to show that any function $W: \mathbb{R}^n \times \mathbb{R}^n \times \mathbb{R}^n  \times \mathbb{R}^n  \rightarrow \mathbb{R}$, such that $W(x,p,y,q) = W(S^T(x,p,y,q))$ for any symplectic transformation $S$ of the form given by Eq. \ref{symplectic}, only depends on the following three parameters: $||x||^2+||p||^2$, $||y||^2+||q||^2$ and $x \cdot y - p \cdot q$.

Our goal is therefore to prove the following: for any pair of quadruples $(x, p, y, q)$ and $(x', p', y', q')$ such that 
\begin{equation}
\label{condition}
\left\{
\begin{array}{ccc}
||x||^2+ ||p||^2=||x'||^2+||p'||^2\\
||y||^2+||q||^2=||y'||^2+||q'||^2\\
x \cdot y - p \cdot q = x' \cdot y' - p' \cdot q',
\end{array}
\right.
\end{equation}
one has: $W(x,p,y,q) = W(x',p',y',q')$. 

Let us introduce the following vectors:
\begin{eqnarray}
a = x + i p &,& a' = x' + i p'\\
b = y - i q &,& b' = y' - i q'.
\end{eqnarray}
The condition \ref{condition} can be rewritten as: 
\begin{equation}
\label{condition2}
\left\{
\begin{array}{ccc}
||a||^2=||a'||^2\\
||b||^2=||b'||^2\\
\mathrm{Re}\langle a | b \rangle = \mathrm{Re}\langle a' | b' \rangle
\end{array}
\right.
\end{equation}
where $\mathrm{Re}(x)$ refers to the real part of $x$.
It is sufficient to prove that there exists an unitary transformation $U \in U(n)$ such that $U a =a'$ and $U b = b'$. Indeed, one can split $U$ into real and imaginary parts: $U = X - iY$, and it is easy to check that $S(X,Y)$ gives the correct change of coordinates. Since $W$ is invariant under this change of coordinates, one concludes that $W(x,p,y,q) = W(x',p',y',q')$. 

Let us introduce the following notations: $A \equiv ||a||^2=||a'||^2, B \equiv ||b||^2=||b'||^2$ and $C \equiv \mathrm{Re}\langle a| b \rangle = \mathrm{Re}\langle a'| b' \rangle$.

Consider first the case where $a$ and $b$ are colinear. This means that $b = C/A a$ and $C = \pm \sqrt{AB}$. 
Using the Cauchy-Schwarz inequality, $|C| = |a' \cdot b'| \leq ||a'|| \cdot||b'|| =  \sqrt{AB}$ with equality if and only if $a'$ and $b'$ are colinear. This means that $a'$ and $b'$ are colinear and that $b' = (C/A) \, a'$. 
Because $||a|| = ||a'||$, the reflexion $U$ across the mediator hyperplane of $a$ and $a'$ is a unitary transformation that maps $a$ to $a'$. This reflexion also maps $b$ to $b'$. This ends the proof in the case where $a$ and $b$ are colinear.

Let us now consider the general case where $a$ and $b$ are not colinear. It is clear that $a'$ and $b'$ cannot be colinear either. We take two bases $(a, b, f_3, \cdots, f_n)$ and $(a', b', f_3', \cdots, f_n')$ of $\mathbbm{C}^n$ and use the Gram-Schmidt process to obtain two orthonormal bases $\mathcal{B}=(e_1, \cdots, e_n)$ and $\mathcal{B}'=(e_1', \cdots, e_n')$. Note that vectors $e_1, e_2, e_1'$ and $e_2'$ are given by:
\begin{eqnarray}
e_1 = \frac{a}{\sqrt{A}} &,& e_2 = \frac{b - \langle e_1|b\rangle e_1}{||b - \langle e_1|b\rangle e_1 ||}\\
e_1' = \frac{a'}{\sqrt{A}} &,& e_2' = \frac{b' - \langle e_1'|b'\rangle e_1'}{||b' - \langle e_1'|b\rangle e_1' ||}.
\end{eqnarray}
Let us call $U$ the unitary operator mapping $\mathcal{B}$ to $\mathcal{B}'$. It is easy to see that $U$ maps $a$ and $b$ to $a'$ and $b'$, respectively. This concludes our proof.

\section{Normal statistical model}
\label{normal_model}

In this section, we discuss briefly the problem of parameter estimation in continuous-variable protocols with a Gaussian modulation. This question is particularly relevant when one is concerned with a finite-size analysis of the security of the protocol (a more detailed presentation can be found in \cite{lev09,LGG10}).

One of the main differences between the asymptotic and the finite-size study of a protocol lies in the parameter estimation. In the former case, one typically assumes that the quantum state $\rho_{AB}$ is known from Alice and Bob while in the latter case, this state needs being estimated. 

For continuous-variable protocols with a Gaussian modulation, it is known that Gaussian attacks are optimal (among collective attacks) and therefore, the secret key rate only depends on the covariance matrix of $\rho_{AB}$. This means that only this covariance matrix, that is, a finite number of parameters,	 needs to be estimated in practice. Moreover, using the symmetries described in this article, one can see that three parameters are in fact sufficient, namely Alice's and Bob's variances, and their covariance. More precisely, the covariance matrix $\Gamma_{AB}$ of the state $\rho_{AB}$ can be assumed to have the following form:
\begin{equation}  
\Gamma_{AB}=
\begin{pmatrix}
X \mathbbm{1}_{2n} & Z \sigma_z \\
Z \sigma_z & Y \mathbbm{1}_{2n}\\
\end{pmatrix}
\end{equation}
with $\sigma_z=\mathrm{diag}(1,-1, 1, -1, \cdots, 1, -1)$.

Furthermore, in a Prepare and Measure implementation of the protocol, $X$ simply corresponds to Alice's modulation variance, which is \emph{a priori} known from Alice and Bob. Hence, only two paramters remain to be estimated in practice.
Asymptotically, this is not a problem since one can assume that the parameter estimation is done perfectly. However, for a finite-size analysis, which is eventually required to prove the security of a practical scheme, it is crucial to have an upper bound on the error in the parameter estimation. Indeed, in an adversarial scenario such as QKD, the legitimate parties should always consider the \emph{worst} covariance matrix compatible with their data except with some small probability $\epsilon$.

This can be easily done once a statistical model is given for the data ${\bf x} = (x_1, \cdots, x_n)$ and ${\bf y} = (y_1, \cdots, y_n)$ observed by Alice and Bob, respectively.

Whereas this could be done even without a model in the case of bounded parameters such as the quantum bit error rate for discrete-variable QKD protocols, this is much more complicated for \emph{a priori} unbounded such as the excess noise in the GG02 protocol.

Then the demonstration given above that the state $\rho_{AB}$ can be considered Gaussian has a  crucial consequence~: since
the classical data ${\bf x}$ and ${\bf y}$ are obtained by performing Gaussian measurements (either homodyne or heterodyne detection), 
the joint distribution of $({\bf x}, {\bf y})$ corresponds to some marginal of a Gaussian Wigner function, and therefore it is also Gaussian.
As a consequence, the variables $x_i$ and $y_i$ (for $i \in \{1, \cdots, n\}$) are related through:
\begin{equation}
y_i = \alpha x_i + z_i,
\end{equation}
where $\alpha$ is a constant and $z_i$ is a Gaussian random variable: $z_i \sim \mathcal{N}(0,\sigma^2)$ which is independent of $x_i$. This is 
the definition of a normal statistical model, where one tries to estimate the values of $\alpha$ and $\sigma^2$. For such a model, one can bound the errors made in the estimation of both $\alpha$ and $\sigma^2$, and therefore on $Y$ and $Z$ (since these are simple functions of $\alpha$ and $\sigma^2$).
Finally, and this is a crucial step in finite-key analysis, one can compute the worst key rate compatible with the data, except with probability $\epsilon$.

\begin{acknowledgments}
We thank Fr\'ed\'eric Grosshans for helpful remarks on a previous version of this work. 
We acknowledge support from Agence Nationale de la Recherche under projects PROSPIQ (ANR-06-NANO-041-05) and SEQURE (ANR-07-SESU-011-01).
\end{acknowledgments}


\end{document}